\documentclass[letterpaper, 10 pt,  conference]{ieeeconf}  

\usepackage{graphicx} 
\usepackage{amsmath,bm,times} 
\usepackage{amssymb}  
\usepackage{tikz}
\usetikzlibrary{shapes,arrows,backgrounds,fit,positioning}
\usepackage{subfigure}
\usepackage{cite}
\usepackage{soul}
\usepackage{balance}
\usepackage{tabularx}
\usepackage{url}

\newtheorem{assumption}{Assumption}

\newtheorem{remark}{Remark}
\newtheorem{theorem}{Theorem}

\newtheorem{lemma}{Lemma}
\newtheorem{example}{Example}

\title{\LARGE \bf Learning Equilibrium with Estimated Payoffs in Population Games}%
\author{Shinkyu Park
  \thanks{This work was supported by funding from King Abdullah University of Science and Technology (KAUST).}
  \thanks{The author is with Electrical and Computer Engineering, King Abdullah University of Science and Technology (KAUST), Thuwal, 23955-6900, Kingdom of Saudi Arabia. {\tt shinkyu.park@kaust.edu.sa}} \\
}
\IEEEoverridecommandlockouts
\begin{document}

\maketitle

\begin{abstract}
  We study a multi-agent decision problem in population games, where agents select from multiple available strategies and continually revise their selections based on the payoffs associated with these strategies. Unlike conventional population game formulations, we consider a scenario where agents must estimate the payoffs through local measurements and communication with their neighbors. By employing task allocation games -- dynamic extensions of conventional population games -- we examine how errors in payoff estimation by individual agents affect the convergence of the strategy revision process. Our main contribution is an analysis of how estimation errors impact the convergence of the agents' strategy profile to equilibrium. Based on the analytical results, we propose a design for a time-varying strategy revision rate to guarantee convergence. Simulation studies illustrate how the proposed method for updating the revision rate facilitates convergence to equilibrium.
\end{abstract}

\section{Introduction} \label{sec:introduction}
Large-population game frameworks provide a systematic approach to studying the decision-making processes of multiple agents engaged in repeated strategic interactions. These frameworks find applications across a wide range of fields, including road congestion \cite{8721089, LI2023110879}, communication systems \cite{5674046, Tembine_2008}, distributed control systems \cite{7823106}, and multi-robot task allocation \cite{9561809, 8439076, 10.1007/978-3-030-20915-5_54}, among others. In this work, we investigate a decision-making model within the population game framework \cite{Sandholm2010-SANPGA-2}, where agents utilize the model to select a strategy among a finite set of available options. This model consists of a \textit{learning rule} (also referred to as a \textit{revision protocol}) and \textit{stochastic alarm clock}, defining how individual agents repeatedly revise their strategy selections, with the primary goal of learning the best strategies.
Under this model, the agent strategy selection is influenced by  information about payoffs of available strategies. Within the standard framework, it is generally assumed that all agents have perfect knowledge of these payoffs.

However, in many engineering applications, such as decentralized control systems \cite{6425819}, where sensing and decision-making are decentralized, this assumption requires well-established communication channels to ensure that all agents have complete knowledge of the payoffs. Otherwise, agents must estimate the payoffs based on local observations and communication with neighbors, and then make decisions on strategy selection based on their own payoff estimates. In this context, limited communication between agents leads to non-negligible estimation errors. This paradigm challenges one of the key assumptions of the standard population game formalism, which is critical for establishing convergence results \cite{Sandholm2010-SANPGA-2}.

In this paper, we adopt task allocation games \cite{9561809, 10383344} and analyze the effect of errors in payoff estimation on the convergence of the agent strategy revision process. These games can be viewed as dynamic extensions of conventional population games, with potential applications in multi-robot systems research as demonstrated in \cite{9561809}. In task allocation games, agents choose strategies to complete a set of predefined tasks, where the payoff for each strategy is determined by the amount of jobs remaining in its associated task. Unlike their conventional counterparts, the game has its own state, governed by a dynamical system model, to represent the remaining jobs for each task.
In such dynamic game settings, if agents have limited capabilities for observing the game's state, they must estimate it and use this estimate to infer expected payoffs and select one of available strategies.

Relevant to our present work where agents are communicating with their neighbors defined by a graph for payoff estimation, the literature on games over graphs, as exemplified by \cite{SZABO200797}, studies the long-term strategic interactions of multiple agents when social networks, represented by graphs, restrict the interactions between them. The study by \cite{BaObQu2017} proposes a new population dynamics framework designed to model distributed information structures in population games.
In this framework, the strategic decision-making of multiple agents is represented by a dynamical system model, concisely expressed as an ordinary differential equation (ODE) on a graph. Using this ODE model, the authors examine convergence of population dynamics.

As part of the series of works, the author of  \cite{2f3e8fce-f2ca-35f7-af3c-454ab8ea0188} proposes a framework to rigorously formulate the agent decision-making process in scenarios where each agent has limited access to information about the underlying game. This limitation results in noisy evaluations of payoffs for available strategies and limited knowledge of other agents' strategy selections. The work also evaluates equilibrium states of the process to predict the long-term behavior of the agents' decision-making. The work of \cite{FOSTER1990219} discusses analytical methods for assessing the long-term behavior of agents' noisy decision-making processes, described by stochastic dynamical system models. Notably, the main results illustrate how a \textit{stochastically stable equilibrium} emerges as the noise level in the agent decision-making process approaches zero.

In our work, we analyze how payoff estimation errors affect the convergence of the agent decision-making process in task allocation games. Unlike previous studies that primarily focus on analyzing the impact of errors on the decision-making process, we propose a method to eliminate this influence in our problem setting.
Specifically, we rigorously prove that with a decreasing strategy revision rate, agents’ decision-making process effectively mitigates the effects of estimation errors, allowing them to asymptotically learn and attain the equilibrium corresponding to the optimal strategy selection. Based on this analysis, we propose a design for a stochastic alarm clock to ensure convergence. Our technical contributions are summarized as follows:
\begin{itemize}
\item Leveraging passivity-based analytical tools \cite{Fox2013Population-Game, 9029756} developed for population games, we analyze the impact of payoff estimation errors on the convergence of the agent decision-making process to its equilibrium state in task allocation games. In particular, we discuss how the influence of the estimation error can be mitigated by decreasing the rate of agents' strategy revision.
  
\item Based on the analysis, we propose the design of a stochastic alarm clock using a non-homogeneous Poisson process to guarantee the convergence. Through simulations with a large number of agents, we illustrate our analytical results and evaluate the effectiveness of the proposed clock design.

\end{itemize}

The paper is organized as follows: In Section~\ref{sec:problem_description}, we formally introduce task allocation games and an agent decision-making model, specifying when and how each agent revises its strategy selection and how it estimates the payoffs for available strategies. In Section~\ref{sec:analysis}, we present an analysis of how errors in payoff estimation affect convergence of the agent decision-making process in task allocation games, and propose a design for a non-homogeneous Poisson alarm clock that guarantees convergence. In Section~\ref{sec:simulations}, we present simulation results that illustrate both our analytical findings and the effectiveness of the proposed clock design.

\section{Preliminaries and Problem Description} \label{sec:problem_description}
In task allocation games \cite{9561809}, there is a finite number $n$ of tasks for a population of agents to carry out. Each agent can select one of available strategies to address the tasks. In this paper, for a concise presentation, we assume that the number of strategies and tasks is the same. We denote the amount of jobs to be completed associated with each task $i \in \{1, \cdots, n\}$ by a non-negative, time-dependent variable $q_i(t) \in [0, q_{\text{max}}]$. We refer to $q(t) = (q_1(t), \cdots, q_n(t))$ as the \textit{game state}. When $q_i(t)$ is below its maximum, $q_{\text{max}}$, the variable increases at a constant rate as more jobs are assigned to the agents and decreases as the agents complete the jobs. This decrease is based on the agents' strategy selections, represented by the \textit{population state} $x(t) = (x_1(t), \cdots, x_n(t))$, where each entry $x_i(t)$ denotes the fraction of the population selecting strategy~$i$. Let $\mathbb X$ be the space of all viable population states, defined as $\mathbb X = \left\{ x \in \mathbb R_+^n \,\big|\, \textstyle \sum_{i=1}^n x_i = 1 \right\}$, where $\mathbb R_+^n$ is the set of (element-wise) non-negative $n$-dimensional vectors.

\subsection{Task Allocation Game Model} We adopt the following dynamical system representation from \cite{9561809} to specify how $q_i(t)$ varies over time: If $q_i(t) < q_{\text{max}}$, then 
\begin{align} \label{eq:task_allocation_games}
  \dot q_i(t) =
    -\underbrace{\mathcal F_i (q_i(t), x_i(t))}_{\text{decrease rate}} + \underbrace{w_i}_{\text{increase rate}}, ~ q_i(0) \geq 0,
\end{align}
and $\dot q_i(t) = \min \{0, -\mathcal F_i (q_i(t), x_i(t)) + w_i \}$ if $q_i(t) = q_{\text{max}}$, where the latter condition ensures that $q_i(t)$ does not exceed its maximum $q_{\text{max}}$.
The function $\mathcal F_i: [0, q_{\text{max}}] \times \mathbb R_+ \to \mathbb R_+$
is non-negative\footnote{Note that even if the $i$-th portion $x_i(t)$ of the population state is within $[0, 1]$, we define the domain of $\mathcal F_i$ corresponding to $x_i(t)$ as the entire set $\mathbb R_+$ of non-negative real numbers.}, and $w_i$ is a positive constant. 
Also, to ensure that $q_i(t)$ is non-negative for all $t \geq 0$, the function $\mathcal F_i(q_i, x_i)$ is zero whenever $q_i = 0$.

According to \eqref{eq:task_allocation_games}, each agent can carry out one task at a time. The output of the game model is the payoff vector $p(t) = (p_1(t), \cdots, p_n(t))$, where we define each $p_i(t)$ as $p_i(t) = q_i(t)$ to associate $p_i(t)$ with the amount of remaining jobs in the $i$-th task. Consequently, under this definition of the payoff vector, agents can be incentivized to carry out the tasks with more remaining jobs.
The following are assumptions we impose on \eqref{eq:task_allocation_games}.
\begin{assumption} \label{assumption:function_F}
  The function $\mathcal F_i: [0, q_{\text{max}}] \times \mathbb R_+ \to \mathbb R_+$ is continuously differentiable and satisfies
  \begin{subequations}
    \begin{align}
      &\lim_{x_i \to \infty} \mathcal F_i(q_i, x_i) = \infty ~ \text{for any positive $q_i$} \label{eq:assumption_a_on_F} \\
      &\frac{\partial \mathcal F_i}{\partial x_i}(q_i, x_i) > 0 ~ \text{for any positive $q_i, x_i$} \label{eq:assumption_b_on_F} \\
      &\frac{\partial \mathcal F_i}{\partial q_i}(q_i, x_i) > 0 ~ \text{for any  positive $q_i, x_i$}. \label{eq:assumption_c_on_F}
    \end{align}
  \end{subequations}
  In other words, the strategies are designed such that \eqref{eq:assumption_a_on_F} and \eqref{eq:assumption_b_on_F} imply that the more agents there are taking on the same task, the faster it can be completed. Additionally, \eqref{eq:assumption_c_on_F} suggests that the larger the amount of remaining jobs, the easier it is for the agents to locate and coordinate to complete them.
\end{assumption}

Note that, as we have proved in Lemma~\ref{lemma:delta_antipassivity} in Appendix~\ref{sec:delta_antipassivity_tag}, under Assumption~\ref{assumption:function_F}, the game model \eqref{eq:task_allocation_games} is \textit{$\delta$-antipassive}, as defined in the appendix.

\begin{assumption} \label{assumption:unique_equilibrium}
  The game model \eqref{eq:task_allocation_games} has a unique equilibrium $(q^\ast, x^\ast)$ in $[0, q_{\text{max}})^n \times \mathbb X$ that satisfies $x_i^\ast (q_j^\ast - q_i^\ast) \leq 0, \, \forall i,j \in \{1, \cdots, n\}$.
\end{assumption}

As detailed in Section~\ref{sec:decision_making_model}, this implies the uniqueness of the equilibrium of \eqref{eq:task_allocation_games} under the agent decision-making model considered in this work.

\begin{example} \label{example:task_allocation_games}
  As demonstrated in \cite{9561809}, the game model \eqref{eq:task_allocation_games} can be applied to a multi-robot trash collection application. In this application, there are $n$ spatially separated patches, with the variable $q_i(t)$ representing the volume of trash in the $i$-th patch. The function $\mathcal F_i$ is defined as
  \begin{align} \label{eq:task_allocation_game_example}
    \mathcal F_i(q_i, x_i) = R_i \frac{e^{\alpha_i q_i} - 1}{e^{\alpha_i q_i} + 1} x_i^{\beta_i},
  \end{align}
  where $R_i$, $\alpha_i$, and $\beta_i$ are positive parameters.
\end{example}

\begin{remark}
  For brevity, we consider the scenario where the number of tasks and strategies is the same, and the payoff vector is defined as $p(t) = q(t)$. Also, as specified in \eqref{eq:task_allocation_games}, every agent can undertake only one task at a time. However, this scenario can be extended to more general cases where the set of available strategies exceeds the number of tasks, and some strategies enable an agent to carry out multiple tasks simultaneously. In such scenarios, the payoff vector needs to be redefined as a non-trivial function of $q(t)$, e.g., $p(t) = G q(t)$. Primary investigations of such extensions have been conducted in \cite{10383344}, and we consider adopting the analysis from this reference as a future research direction.
\end{remark}

\subsection{Learning Rule, Stochastic Alarm Clock, and Evolutionary Dynamics Model} \label{sec:decision_making_model}
Two central components of the agent decision-making model are the \textit{learning rule} (also referred to as the \textit{revision protocol}) and the \textit{stochastic alarm clock}. The learning rule describes how each agent changes its strategy selection when given an opportunity and is typically defined as a function
\begin{align} \label{eq:revision_protocol}
  \rho_{ji}(p, x) = \mathbb P \left( \text{agent switching strategy from $j$ to $i$} \right).
\end{align}
We consider the class of learning rules that depend only on $p$, i.e., $\rho_{ji}(p, x) = \rho_{ji}(p)$, and are Lipschitz continuous. That is, there exists a positive constant $c$ for which the following inequality holds:
\begin{align} \label{eq:lipschitz_continuity_revision_protocol}
  \left| \rho_{ji}(p) - \rho_{ji} (\bar p) \right| \leq c \left \| p - \bar p \right \|_2, ~ \forall p, \bar p \in \mathbb R^n.
\end{align}
Below is an example of an existing model that belongs to this class.
\begin{example} \label{example:learning_rules}
  Suppose $\rho_{ji}$ is the Smith learning rule, i.e., $\rho_{ji}(p) = \varrho [p_i - p_j]_+ = \varrho \max (0, p_i - p_j)$ for $i \neq j$, originally investigated in transportation research \cite{10.2307/25768135}. Given that the range of each $p_i$ is bounded, we select a constant $\varrho$ to ensure that $\sum_{i=1}^n \varrho [p_i - p_j]_+ \leq 1$ holds for all $j$ in $\{1, \cdots, n\}$.
  We can derive the following inequality and verify the Lipschitz continuity of the Smith learning rule:
  \begin{align}
    \left | \rho_{ji}(p) - \rho_{ji}(\bar p) \right | \leq \sqrt{2} \varrho \left \| p - \bar p \right \|_2.
  \end{align}
\end{example}

While the learning rule defines how each agent revises its strategy, the stochastic alarm clock determines when the agent can make strategy revision using the learning rule. Typically, a homogeneous Poisson process $N(t)$ is utilized to define the clock \cite[Chapter~10]{Sandholm2010-SANPGA-2}. Specifically, at each ring of the clock, defined as any time $t$ satisfying $N(t) - N(t-\delta) \geq 1, ~ \forall \delta>0$, the agent retains the opportunity to revise its strategy. As the Poisson processes assigned to the agents are independent, if they can assess the payoff vector $p$, the agents' strategy revision can be conducted in a decentralized manner.

Suppose these Poisson processes are identically distributed, and let $\lambda$ define the rate of the processes. Consider the following ordinary differential equation:
\begin{align} \label{eq:edm}
  \dot x_i(t) \!=\! \textstyle \lambda \!\sum_{j=1}^n \!x_j(t) \rho_{ji}(p(t)) \!-\! \lambda x_i(t) \!\sum_{j=1}^n \!\rho_{ij}(p(t)).
\end{align}
We refer to \eqref{eq:edm} as the \textit{evolutionary dynamics model (EDM)}. 
It has been well-documented in \cite[Chapters~5 and 10]{Sandholm2010-SANPGA-2} that when the Poisson processes are independent and identically distributed and there is a sufficiently large number of agents, the solution of \eqref{eq:edm} serves as a good predictor of the population state with arbitrarily high accuracy. Throughout this work, we assume that \eqref{eq:edm} is \textit{$\delta$-passive}, as defined in Appendix~\ref{sec:delta_passivity_edm}. We make the following assumptions regarding \eqref{eq:edm}, which is widely known as \textit{Nash stationarity} \cite{9029756}.
\begin{assumption} \label{assumption:nash_stationarity}
  With $p=q$, the following two conditions are equivalent.
  \begin{enumerate}
  \item $\sum_{j=1}^n x_j \rho_{ji}(p) \!-\! x_i \sum_{j=1}^n \rho_{ij}(p) \!=\! 0, \forall i \in \{1, \cdots, n\}$ \vspace{.1cm}
  \item $x_i (q_j - q_i) \leq 0, \, \forall i,j \in \{1, \cdots, n\}$
  \end{enumerate}
  \vspace{.1cm}
\end{assumption}
Note that when $\rho_{ji}$ is defined as the Smith learning rule, \eqref{eq:edm} satisfies Assumption~\ref{assumption:nash_stationarity}. Furthermore, in conjunction with Assumption~\ref{assumption:unique_equilibrium}, Assumption~\ref{assumption:nash_stationarity} implies that the feedback interconnection of \eqref{eq:task_allocation_games} and \eqref{eq:edm}, has a unique equilibrium.

In addition, according to Assumptions~\ref{assumption:unique_equilibrium} and \ref{assumption:nash_stationarity}, the equilibrium state $(q^\ast, x^\ast)$ satisfies $\mathcal F_i(q_i^\ast, x_i^\ast) = w_i, ~\forall i \in \{1, \cdots, n\}$ and $q_i^\ast = q_j^\ast,~\forall i,j \in \{1, \cdots, n\}$. Consequently, we can infer that $(q^\ast, x^\ast)$ is the optimal state at which the infinity norm $\|q(t)\|_{\infty}$ of the game state is minimized over the set $\mathbb O$ of stationary points of \eqref{eq:task_allocation_games}, defined by $\mathbb O = \{(q,x) \in [0, q_{\text{max}}]^n \times \mathbb X \,|\, \mathcal F_i (q_i, x_i) = w_i, ~ \forall i \in \{1, \cdots, n\}\}$. In other words, it holds that $\|q^\ast\|_{\infty} = \min_{(q,x) \in \mathbb O} \|q \|_{\infty}$.

\subsection{Payoff Vector Estimation}
In our problem formulation, instead of directly assessing $q(t)$, each agent~$k$ observes a function $y^{(k)}(t) = h^{(k)}(q(t))$ of the game state $q(t)$ and communicates with its neighbors to estimate $q(t)$. For instance, some agents may observe the full state $q(t)$, others might only take measurements of $q_i(t)$ associated with their strategy selection $i$, or yet others may not be able to collect any measurements of $q(t)$ at all. Since not every agent in the population can directly observe the full state $q(t)$, they adopt an estimation rule to infer the full state using their own observations $y^{(k)}(\tau), \tau \in [0,t]$ and information from their neighbors.

Motivated by the large literature on distributed state estimation \cite{REGO201936},
we consider that each agent~$k$ shares its own estimate $\hat q^{(k)}(t)$ of $q(t)$ whenever it can communicate with its neighbors. Let $\mathbb N_k (t)$ denote the set of neighbors of agent~$k$ at time~$t$. Given its observation $y^{(k)}(t)$ of $q(t)$ and estimates $\{ \hat q^{(l)}(t^-) \}_{l \in \mathbb N_k(t)}$ from its neighbors\footnote{We use the notation $\hat q^{(l)}(t^-)$ to denote the estimate of agent~$l$ specifically before the agent updates its estimate at time~$t$.}, the agent updates its estimate $\hat q^{(k)}(t)$ from which the estimate $\hat p^{(k)}(t)$ of the payoff vector $p(t)$ can be derived as $\hat p^{(k)}(t) = \hat q^{(k)} (t)$. We represent the estimation rule as
\begin{align} \label{eq:state_estimation_rule}
  \hat q^{(k)}(t) = g \Big( \big \{ \hat q^{(l)}(t^-) \big \}_{l \in \mathbb N_k(t)}, y^{(k)}(t) \Big).
\end{align}
We provide the following example to illustrate this.

\begin{example} \label{eq:payoff_estimation_example}
  Given $N$ agents in the population, suppose that only $N_{\text{leader}} (< N)$ leader agents can observe the game state $q(t)$, while the others cannot. Assuming that the agents are not aware of the game model \eqref{eq:task_allocation_games},
  the estimation rule \eqref{eq:state_estimation_rule} can be implemented as
\begin{align} \label{eq:consensus_estimation_rule}
  \hat q^{(k)} (t) \!=\!
  \begin{cases}
    q(t)\! & \text{if $k$ is a leader} \\
    \frac{1}{|\mathbb N_k(t)|} \sum_{l \in \mathbb N_k(t)} \hat q^{(l)}(t^-)\! & \text{otherwise}.
    \end{cases}
\end{align}
In other words, agent~$k$ sets its estimate $\hat q^{(k)}(t)$ to $q(t)$ if it is a leader and directly measures $q(t)$. Otherwise, the agent updates $\hat q^{(k)}(t)$ to the average of its neighbors' estimates. 
\end{example}

The main results of this paper are applicable to any estimation rule $g$, provided it satisfies the following two assumptions:
1) Given that the game state $q(t)$ is bounded, the estimation error $\hat q^{(k)}(t) - q(t)$ also remains bounded for all $t \geq 0$, and 2) as the variation in $q(t)$ diminishes, i.e., $\| \dot q(t) \|_2 \to 0$ as $t \to \infty$, all agents in the population can asymptotically recover the full state of $q(t)$. We formally state these assumptions as follows:
\begin{assumption} \label{assumption:estimation_convergence}
  The estimation rule \eqref{eq:state_estimation_rule} satisfies the following conditions for every agent $k$:
  \begin{subequations}
    \begin{align}
      &\sup_{t \geq 0} \epsilon^{(k)} (t) < B_{\epsilon} \\
      &\lim_{t \to \infty} \left\| \dot q(t) \right\|_2 = 0 \implies \lim_{t \to \infty} \epsilon^{(k)} (t) = 0,
    \end{align}
  \end{subequations}
  where $ \epsilon^{(k)} (t) = \|\hat q^{(k)}(t) - q(t)\|_2$ and $B_{\epsilon}$ is a fixed positive constant.
\end{assumption}

Note that when the underlying communication graph is fixed and strongly connected, \eqref{eq:consensus_estimation_rule} satisfies Assumption~\ref{assumption:estimation_convergence}.


\section{Convergence Analysis and Revision Rate Update} \label{sec:analysis}
We discuss how the estimation error $\epsilon^{(k)}(t)$ influences convergence of the population state and the game state. In a population of $N$ agents, suppose each agent~$k$ adopts the learning rule $\rho_{ji}(\hat p^{(k)}(t))$ using its own payoff vector estimate $\hat p^{(k)}(t)$ for strategy revision. Let $\mathbb M_j^N (t)$ denote the set of agents adopting strategy~$j$ at time~$t$. According to the definitions of the learning rule and the Poisson alarm clock, when the clock of an agent rings, the probability of that agent currently choosing strategy~$j$ and switching to strategy~$i$ can be specified as 
\begin{align} \label{eq:average_strategy_revision_probability_finite}
  x_j^N(t) \sum_{k \in \mathbb M_j^N(t)} \frac{\rho_{ji}(\hat p^{(k)} (t))}{|\mathbb M_j^N(t)|},
\end{align}
where $x_j^N(t)$ is the fraction of the $N$-agent population adopting strategy~$j$ at time~$t$, $\hat p^{(k)}(t) = \hat q^{(k)}(t)$ is agent~$k$'s estimate of the payoff vector $p(t)$, and $|\mathbb M_j^N(t)|$ denotes the cardinality of the set $\mathbb M_j^N(t)$.
Note that if $\mathbb M_j^N(t)$ is empty, implying $x_j^N(t) = 0$, the expression in  \eqref{eq:average_strategy_revision_probability_finite} is considered zero.

When the size of the population becomes arbitrarily large, that is, as $N$ tends to infinity, we can derive
\begin{multline} \label{eq:average_strategy_revision_probability}
  \lim_{N \to \infty} x_j^N(t) \!\sum_{k \in \mathbb M_j^N(t)} \!\frac{\rho_{ji}(\hat p^{(k)} (t))}{|\mathbb M_j^N(t)|} \\ = x_j(t) \lim_{N \to \infty} \sum_{k \in \mathbb M_j^N(t)} \frac{\rho_{ji}(\hat p^{(k)} (t))}{|\mathbb M_j^N(t)|},
\end{multline}
where $x_j(t)$ represents the fraction of agents in the infinite population adopting strategy~$j$. The limit $\lim_{N \to \infty} \sum_{k \in \mathbb M_j^N(t)} \frac{\rho_{ji}(\hat p^{(k)} (t))}{|\mathbb M_j^N(t)|}$ in the right-hand side of \eqref{eq:average_strategy_revision_probability} represents the average strategy revision probability of all $j$-strategists. Assuming this limit exists and based on the definition of $\rho_{ji}$ in \eqref{eq:revision_protocol}, we can validate that
\begin{enumerate}
\item $0 \leq \lim_{N \to \infty} \sum_{k \in \mathbb M_j^N(t)} \frac{\rho_{ji}(\hat p^{(k)} (t))}{|\mathbb M_j^N(t)|} \leq 1$, and

\item $\sum_{i=1}^n \lim_{N \to \infty} \sum_{k \in \mathbb M_j^N(t)} \frac{\rho_{ji}(\hat p^{(k)} (t))}{|\mathbb M_j^N(t)|} = 1$.
\end{enumerate}

Consequently, analogous to the derivation of \eqref{eq:edm}, when agents revise their strategies based on estimates of the payoff vector and using Poisson alarm clocks with a rate of $\lambda$, the following EDM can be derived:
\begin{align} \label{eq:edm_with_error_term}
  \dot x_i(t)
  &= \lambda \sum_{j=1}^n x_j(t) \lim_{N \to \infty} \sum_{k \in \mathbb M_j^N(t)} \frac{\rho_{ji}(\hat p^{(k)} (t))}{|\mathbb M_j^N(t)|} \nonumber \\
  & \quad - \lambda x_i(t) \sum_{j=1}^n \lim_{N \to \infty} \sum_{k \in \mathbb M_i^N(t)} \frac{\rho_{ij}(\hat p^{(k)} (t))}{|\mathbb M_i^N(t)|} \nonumber \\
  &=\lambda \!\sum_{j=1}^n x_j(t) \rho_{ji}(p(t)) - \lambda x_i(t) \!\sum_{j=1}^n \rho_{ij}(p(t)) + \lambda \xi_i(t),
\end{align}
where $\xi_i(t)$ is defined as
\begin{align*}
  \xi_i(t)
  &= \sum_{j=1}^n x_j(t) \left( \lim_{N \to \infty} \!\sum_{k \in \mathbb M_j^N(t)} \!\frac{\rho_{ji}(\hat p^{(k)} (t))}{|\mathbb M_j^N(t)|} - \rho_{ji}(p(t)) \right) \nonumber \\
  & \quad \!-\! x_i(t) \!\sum_{j=1}^n \!\left(\! \lim_{N \to \infty} \!\sum_{k \in \mathbb M_i^N(t)} \!\frac{\rho_{ij}(\hat p^{(k)} (t))}{|\mathbb M_i^N(t)|} \!-\! \rho_{ij}(p(t)) \!\right).
\end{align*}
According to its definition, $\xi_i(t)$ can be interpreted as the difference between the probability of agents' strategy revision based on the payoff vector $p(t)$ and the probability based on the estimates $\hat p^{(k)}(t)$.

By the Lipschitz continuity of $\rho_{ji}$, we compute the upper bound of $\xi_i(t)$ as
\begin{align} \label{eq:bound_on_edm_mismatch_term}
  | \xi_i(t) |
  &\leq c \sum_{j=1}^n x_j(t) \lim_{N \to \infty} \sum_{k \in \mathbb M_j^N(t)} \frac{\epsilon^{(k)} (t)}{|\mathbb M_j^N(t)|}.
\end{align}
Recall that $\epsilon^{(k)} (t)$ is defined as $\epsilon^{(k)} (t) = \|\hat q^{(k)} (t) - q(t)\|_2$. According to  Assumption~\ref{assumption:estimation_convergence}, since $\epsilon^{(k)} (t)$ is bounded, so does $\xi_i(t)$. Also, we can infer that $\xi_i(t)$ vanishes as does the estimation error $\epsilon^{(k)} (t)$ at every agent~$k$.

\subsection{Convergence to Equilibrium State}
\begin{figure}
  \center
  \includegraphics[trim={0.0in 0.1in 0.0in 0.0in}, clip ,width=3.0in]{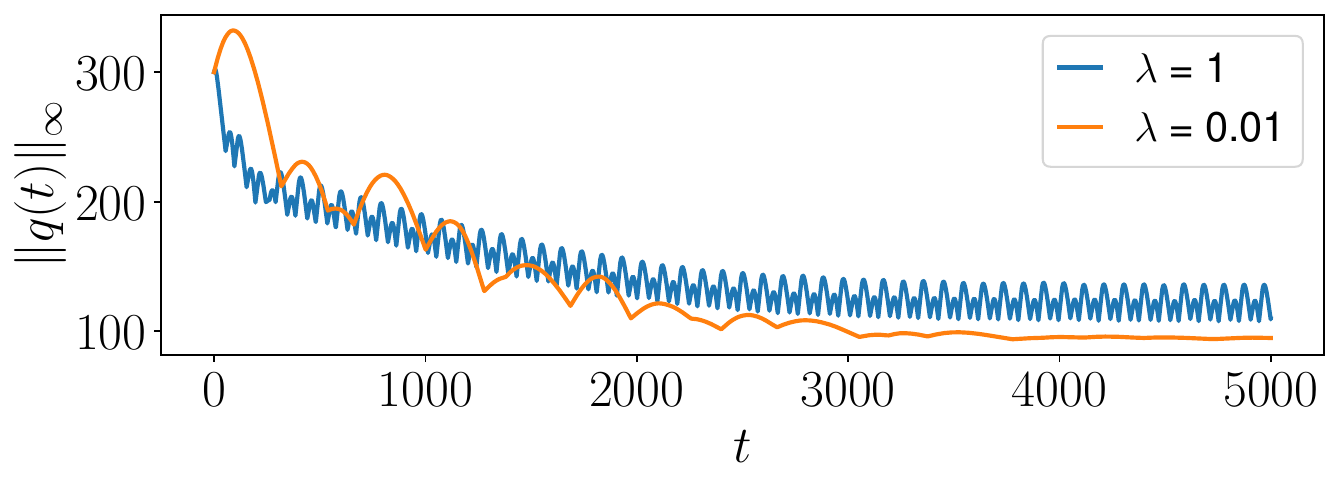}
  \caption{Trajectories of the game state derived by the Smith learning rule in the task allocation game described in Example~\ref{example:task_allocation_games}, where the agents are estimating the payoff vector using \eqref{eq:consensus_estimation_rule}.}
  \label{fig:effect_of_estimation_error}
  \vspace{-1.5em}
\end{figure}
To visualize the effect of the estimation error on the convergence, we conduct simulations using the scenario explained in Example~\ref{example:task_allocation_games}, the Smith learning rule, explained in Example~\ref{example:learning_rules}, and the estimation rule \eqref{eq:consensus_estimation_rule}. As we use the same simulation setup as described in Section~\ref{sec:simulations}, except $\lambda$ is fixed here, we refer to that section for details of the simulation setup.

The simulations are executed with two different revision rates, $\lambda = 0.01$ and $1.0$, and the results are depicted in Fig.~\ref{fig:effect_of_estimation_error}.
When the revision rate is high, $\| q(t) \|_\infty$ experiences a steep decrease early on; however, in the long run, its value tends to oscillate and be higher compared to the scenario with a lower revision rate. Conversely, when the revision rate is low, the agents have fewer opportunities to revise their strategies, resulting in a larger overshoot in the early stages.
Based on these observations, to ensure rapid convergence in the initial stages, the rate should be set sufficiently high; however, to achieve a lower value of $\| q(t) \|_\infty$ in the long term, the rate should be reduced.

Therefore, the optimal design of the strategy revision rate would necessitate a decrease over time.
To provide a rigorous justification for our observation, we present the following theorem. The proof of the theorem is provided in Appendix~\ref{sec:theorem_proof}.
\begin{theorem} \label{thm:convergence}
  For a given revision rate $\lambda$, let $q_{\lambda}(t)$ and $x_{\lambda}(t)$ be the game and population states, respectively, determined by the feedback interconnection of the game model \eqref{eq:task_allocation_games} and EDM \eqref{eq:edm_with_error_term}. Under Assumptions~\ref{assumption:function_F}-\ref{assumption:estimation_convergence}, it holds that
  \begin{align*}
    \limsup_{t \to \infty} \left( \left\| q_{\lambda}(t) - q^\ast \right\|_2 + \left\| x_\lambda(t) - x^\ast \right\|_2 \right) \to 0 \text{ as } \lambda \to 0,
  \end{align*}
  where $(q^\ast, x^\ast)$ is the unique equilibrium state of the closed loop system defined by \eqref{eq:task_allocation_games} and \eqref{eq:edm} with $p(t) = q(t)$.
\end{theorem}

\subsection{Revision Rate Update} \label{sec:methods}
As we stated in Theorem~\ref{thm:convergence} and observed from the simulation results depicted in Fig.~\ref{fig:effect_of_estimation_error}, starting with a high initial revision rate $\lambda$, followed by its proper regulation, ensures the convergence of $(q(t), x(t))$ to the equilibrium state, while particularly achieving steep convergence at the early stages of the game.
The proof of the theorem utilizes the so-called \textit{storage functions} $\mathcal L(q(t), x(t))$ and $\mathcal S(p(t), x(t))$ of the game model \eqref{eq:task_allocation_games} and EDM \eqref{eq:edm}, respectively, which are defined in Appendix~\ref{sec:passivity}, to construct a Lyapunov candidate function. Notably, by decreasing $\lambda$, we establish that both the storage functions diminish over time, thereby ensuring the convergence to the equilibrium state. This process requires knowledge of the game model~\eqref{eq:task_allocation_games} and its associated storage function $\mathcal L(q(t), x(t))$.
However, if the game model \eqref{eq:task_allocation_games} is unknown to the agents, this method of updating $\lambda$ using the game model becomes infeasible.

Instead, we propose updating $\lambda$ when the frequency of the agents revising to other strategies decreases, despite receiving revision opportunities. To implement this, we consider a time-varying revision rate $\lambda(t)$. Let $\{t_m\}_{m=1}^\infty$ and $\{\lambda_m\}_{m=1}^\infty$ be sequences of time instants and rates, respectively, where $\lambda(t)$ of each agent's Poisson alarm clock is defined as $\lambda(t) = \lambda_m$ for $t \in [t_m, t_{m+1})$.
An agent with the most accurate estimates of $(q(t), x(t))$, such as the leader agents in Example~\ref{eq:payoff_estimation_example}, 
updates the revision rate to $\lambda_{m+1} = \gamma \lambda_m$ with $\gamma \in (0, 1)$ at time $t_{m+1}$ if the following condition is met:
\begin{align} \label{eq:revision_rate_update_condition_03}
  \nabla_x^T \mathcal S(p(t_{m+1}), x(t_{m+1})) \mathcal V(p(t_{m+1}), x(t_{m+1})) \geq - \epsilon,
\end{align}
where $\mathcal V = (\mathcal V_1, \cdots, \mathcal V_n)$ with $\mathcal V_i(p,x) = \sum_{j=1}^n x_j \rho_{ji}(p) - x_i \sum_{j=1}^n \rho_{ij}(p)$ and $\epsilon$ is a small positive constant.
It then broadcasts the updated rate $\lambda_{m+1}$ to other agents using the same communication graph employed for payoff vector estimation. According to \eqref{eq:delta_passivity_requirement_02}, with small $\epsilon$, \eqref{eq:revision_rate_update_condition_03} implies that $\mathcal V(p(t_{m+1}), x(t_{m+1}))$ becomes small. Consequently, the agents change their strategies less frequently.

Additionally, we require $t_{m+1} - t_m \geq \frac{\tau}{\lambda_m}$ to ensure that every agent receives a strategy revision opportunity with a certain probability dependent on a constant $\tau$ before the revision rate is updated. To understand this, by definition, $1-e^{-\tau}$ represents the probability of a revision opportunity occurring within the time interval $[t_m, t_m + \frac{\tau}{\lambda_m} )$. Therefore, the inequality requirement suggests that with a larger $\tau$, it becomes increasingly likely that every agent will receive a revision opportunity before the revision rate updates at $t_{m+1}$.

\section{Simulations} \label{sec:simulations}

\begin{figure*}
  \center
  \subfigure[]{
    \includegraphics[trim={0.0in 0.0in 0.0in 0.0in}, clip ,width=1.9in]{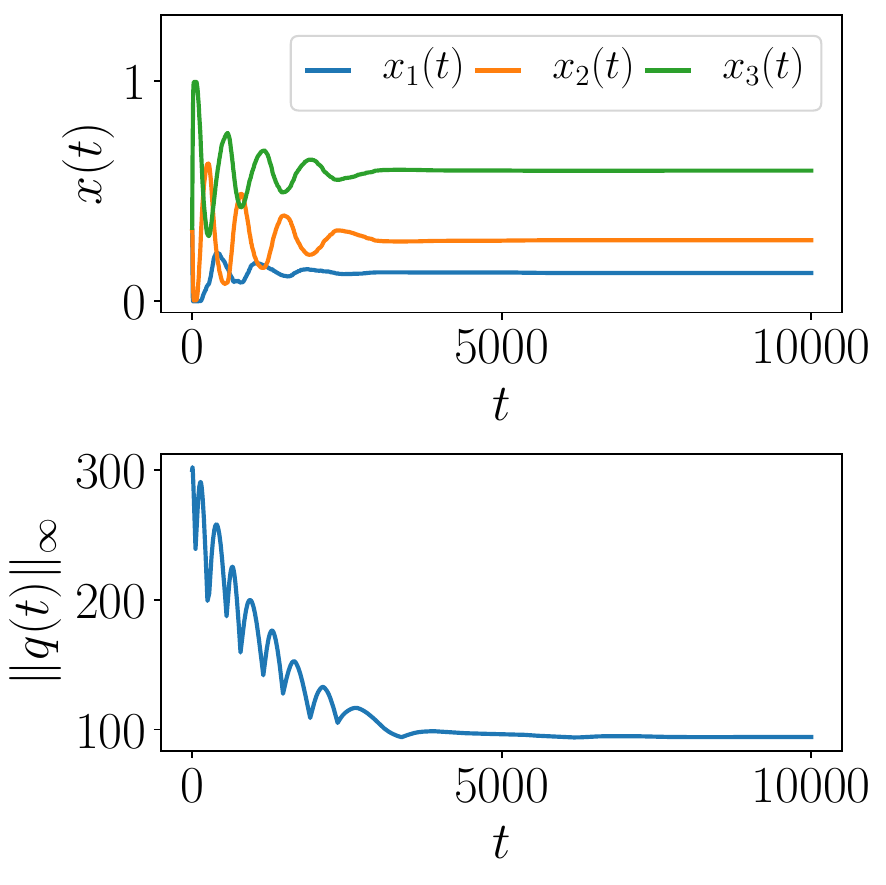}
    \label{fig:simulation_results_a}
  }
  \subfigure[]{
    \includegraphics[trim={0.0in 0.0in 0.0in 0.0in}, clip ,width=1.9in]{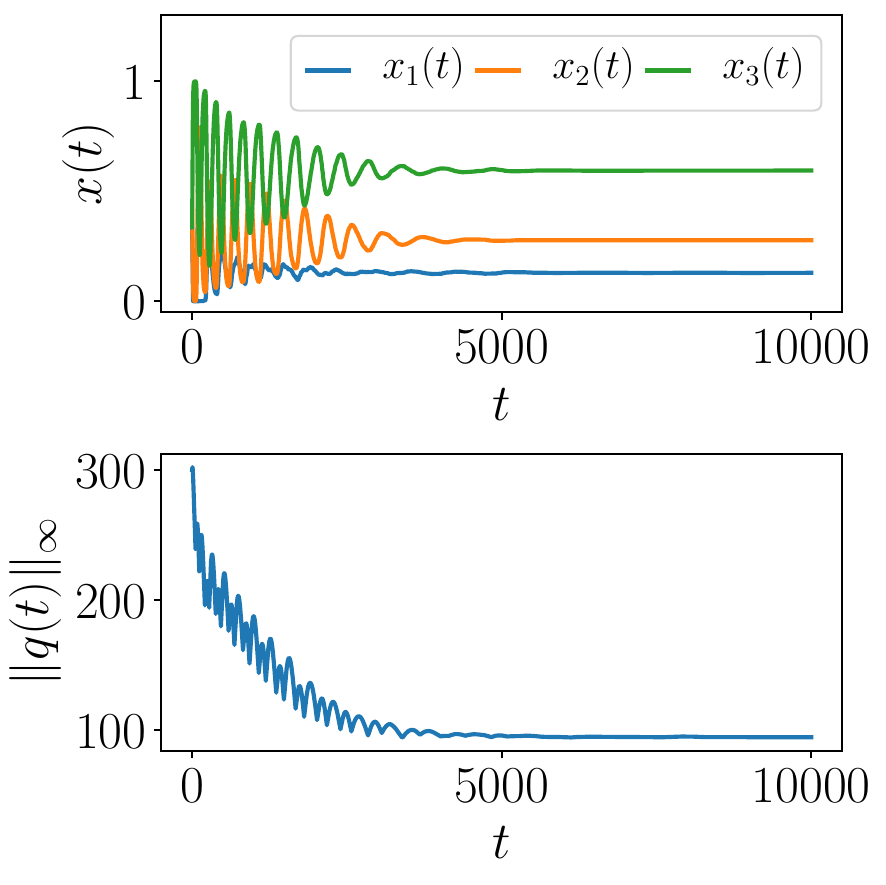}
    \label{fig:simulation_results_b}
  }
  \subfigure[]{
    \includegraphics[trim={0.0in 0.0in 0.0in 0.0in}, clip ,width=1.9in]{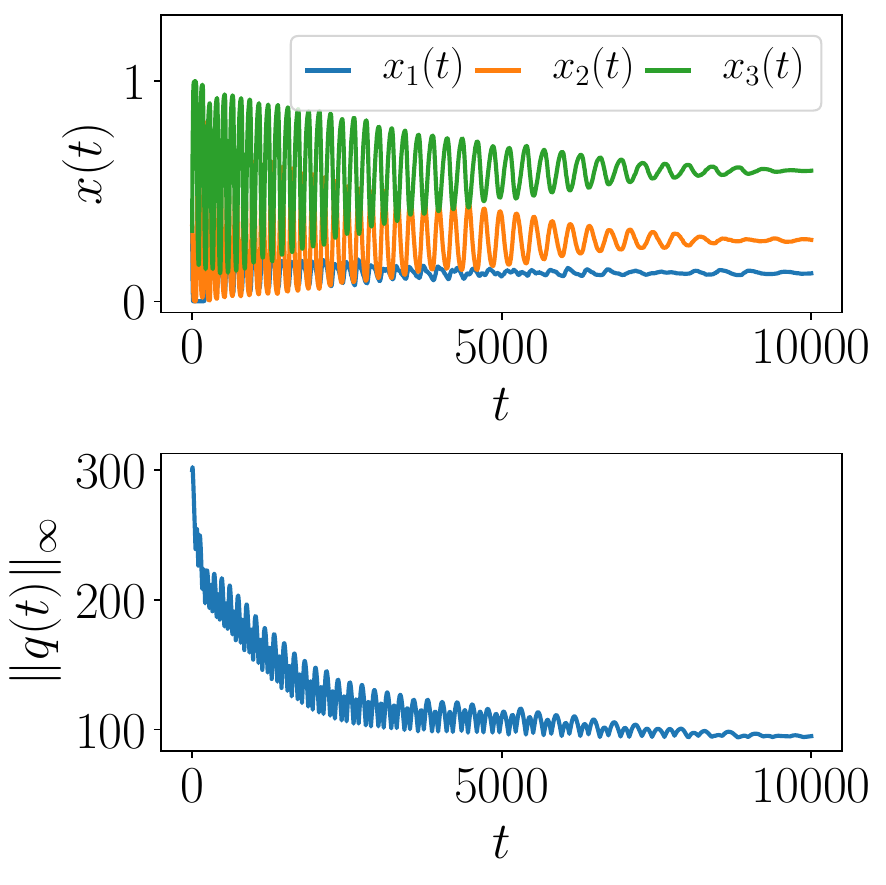}
    \label{fig:simulation_results_c}
  }
  \caption{Population state and game state trajectories when the rate of the Poisson alarm clock is updated according to the method described in Section~\ref{sec:methods}. The trajectories are examined using three different parameter choices: (a) $\gamma=0.8, \tau=0.2$, (b) $\gamma=0.95, \tau=1.0$, and (c) $\gamma=0.99, \tau=1.4$.}
  \label{fig:simulation_results}
  \vspace{-1em}
\end{figure*}

To illustrate our analysis and also to validate the strategy revision rate update method, we design and carry out simulations based on Example~\ref{example:task_allocation_games} with $n=3$.
The parameters of \eqref{eq:task_allocation_game_example} are set as $R_1=R_2=R_3 = 3.44$, $\alpha_1 = \alpha_2 = \alpha_3 = 0.036$, and $\beta_1 = \beta_2 = \beta_3 = 0.91$. The increase rate $w = (w_1, w_2, w_3)$ is specified as $w = (0.5, 1, 2)$. It can be verified that $\mathcal F_i$ satisfies Assumption~\ref{assumption:function_F}. Additionally, by setting a sufficiently large value for $q_{\text{max}}$, we can validate that Assumption~\ref{assumption:unique_equilibrium} holds.

In simulation, there are $3000$ agents in the population and they can share their estimates of the payoff vector via a strongly connected graph, generated by the Erd{\H o}s-R{\' e}nyi model with the edge formation probability $p=0.1$. Among the entire population, we uniformly randomly select $10\%$ of the total population as the leaders who can directly observe the game state $q(t)$, and the rest of the population cannot observe the payoff vector at all, as in the scenario explained in Example~\ref{eq:payoff_estimation_example}. Each agent $k$ computes its own estimate $q^{(k)}(t)$ according to the estimation rule described in \eqref{eq:consensus_estimation_rule}, and revises its strategy using the Smith learning rule $\rho_{ji}(p^{(k)}(t)) = \varrho [p_i^{(k)}(t) - p_j^{(k)}(t)]_+$ with $\varrho = 1/400$.\footnote{The choice $\varrho=1/400$ is made to ensure that $\varrho \sum_{i=1}^n [p_i^{(k)}(t) - p_j^{(k)}(t)]_+ \leq 1$ throughout the simulations.}

The communication and observation of the game state by the agents occur at discrete time instants, specifically at $t = 1, 2, 3, \cdots$. For its strategy revision, each agent takes samples of time instants based on its Poisson alarm clock, and at each sampled time, it revises its strategy using the Smith learning rule and estimated payoff vector.

We iterate the simulation over a range of parameters for the strategy revision update method, fixing the initial conditions to $x(0) = (1/3, 1/3, 1/3)$ and $q(0) = (100, 200, 300)$, and setting the initial estimates of $q(0)$ by all agents to $\hat q^{(k)}(0) = (0, 0, 0)$. Fig.~\ref{fig:simulation_results} illustrates the simulation results for three cases: (a) $\gamma = 0.8, \tau = 0.2$, (b) $\gamma = 0.95, \tau = 1.0$, and (c) $\gamma = 0.99, \tau = 1.4$, with the initial revision rate $\lambda(t_0)$ set to $1$ in all three cases. We fix $\epsilon = 0.01$ as the left-hand side term in \eqref{eq:revision_rate_update_condition_03} approaches zero over time in the simulation. In all cases, the revision rates decrease over time, and we can observe the convergence of the population and game states to the equilibrium state, as we have discussed in Theorem~\ref{thm:convergence}. When the parameters $\gamma$ and $\tau$ are small, the revision rate decreases fast, resulting in the substantial overshoot in the game state trajectory, as observed in Fig.~\ref{fig:simulation_results_a}. As those parameters increase, the game state tends to have a steeper decrease in the early stage, as observed in Fig.~\ref{fig:simulation_results_b}. However, if these parameters are too large, the reduction of the revision rate becomes slow and convergence to the equilibrium takes longer, as illustrated in Fig.~\ref{fig:simulation_results_c}.

\section{Conclusion} \label{sec:conclusions}
We investigated a population game problem in which decision-making agents need to estimate the payoff vector for their strategy selections. By adopting task allocation games, we examined the effect of errors in the payoff vector estimation on the convergence to the equilibrium state. Our analysis showed how the adaptive rate of the agents' strategy revisions mitigates the impact of these errors on the convergence. Leveraging the analytical results, we proposed a design for time-varying strategy revision rates to ensure convergence. As a future direction of research, we plan to explore optimization of the time-varying revision rate design. Also, analyzing how the topology of the communication graph affects the convergence and investigating the optimal communication topology design are topics of our interest.

\appendix
\subsection{Passivity of Game Model and EDM} \label{sec:passivity}
The notions of $\delta$-antipassivity and $\delta$-passivity are adopted from \cite{9029756}.

\subsubsection{$\delta$-Antipassivity of Game Model \eqref{eq:task_allocation_games}} \label{sec:delta_antipassivity_tag}
We call \eqref{eq:task_allocation_games} \textit{$\delta$-antipassive} if there is a continuously differentiable non-negative function $\mathcal L: (0, q_{\text{max}})^n \times \mathbb X \to \mathbb R_+$ satisfying
\begin{subequations} \label{eq:delta_antipassivity_requirements_01}
  \begin{align}
    &\nabla_x \mathcal L(q,x) =  \mathcal F(q,x) - w \\
    &\nabla_q^T \mathcal L(q,x) \left( -\mathcal F(q,x) + w \right) \leq 0,
  \end{align}
\end{subequations}
and
\begin{align} \label{eq:delta_antipassivity_requirements_02}
  \mathcal L(q,x)=0 &\iff \mathcal F(q,x) = w \nonumber \\
  &\iff \nabla_q^T \mathcal L(q,x) \left( -\mathcal F(q,x) + w \right) = 0,
\end{align}
where $\mathcal F = (\mathcal F_1, \cdots, \mathcal F_n)$ and $w = (w_1, \cdots, w_n)$ are defined in \eqref{eq:task_allocation_games}. We refer to $\mathcal L$ as the \textit{$\delta$-antistorage function}.

Adopting the analysis presented in \cite{10383344}, we 
let
\begin{align} \label{eq:delta_antistorage_function}
  \mathcal L(q,x) = \sum_{i=1}^n \int_{x_i^\ast(q_i)}^{x_i} (\mathcal F_i(q_i, s) - w_i) \, \mathrm d s,
\end{align}
where $x_i^\ast: (0, q_{\text{max}}) \to \mathbb R_+$
is a function that maps the game state $q_i$ to an element in $\mathbb R_+$ satisfying $\mathcal F_i(q_i, x_i^\ast(q_i)) = w_i$. 
Note that by Assumption~\ref{assumption:function_F}, for a given positive $q_i$, there is a unique $x_i^\ast(q_i)$, and by the implicit function theorem, the function $x_i^\ast$ is continuously differentiable.


\begin{lemma} \label{lemma:delta_antipassivity}
  The function $\mathcal L$ in \eqref{eq:delta_antistorage_function} satisfies \eqref{eq:delta_antipassivity_requirements_01} and \eqref{eq:delta_antipassivity_requirements_02}.
\end{lemma}
\begin{proof}
  Note that $\mathcal L$ can be re-written as
  \begin{align} \label{eq:antidelta_storage_function}
    \mathcal L(q,x)
    &= \sum_{i=1}^n (f_i(q_i, x_i) - w_i x_i),
  \end{align}
  where $f_i(q_i, x_i) = \int_{x_i^\ast(q_i)}^{x_i} \mathcal F_i(q_i, s) \, \mathrm ds + w_i x_i^\ast (q_i)$,
  from which we can infer that $\frac{\partial \mathcal L(q,x)}{\partial x_i} = \mathcal F_i(q_i, x_i) - w_i$.
  Since $s \geq x_i^\ast(q_i)$ if and only if $\mathcal F_i(q_i, s) \geq  w_i$, and $\mathcal F_i(q_i, s)$ is increasing in $s$, we can establish $\mathcal L(q,x) \geq 0$ where the equality holds if and only if $\mathcal F_i(q_i, x_i) = w_i$ for all $i$ in $\{1, \cdots, n\}$.
  
  Notice that
  \begin{align}
    \frac{\partial \mathcal L (q,x)}{\partial q_i} = \frac{\partial f_i(q_i, x_i)}{\partial q_i} = \int_{x_i^\ast(q_i)}^{x_i} \frac{\partial \mathcal F_i(q_i, s)}{\partial q_i} \, \mathrm ds,
  \end{align}
  and since $\mathcal F_i(q_i, x_i)$ is an increasing function of $q_i$, we can derive
  \begin{align}
    &\nabla_q^T \mathcal L(q,x) \left( -\mathcal F(q,x) + w \right) \nonumber \\
    &= \sum_{i=1}^n \frac{\partial \mathcal L (q,x)}{\partial q_i} \left( -\mathcal F_i(q_i, x_i) + w_i \right) \nonumber \\
    &= \sum_{i=1}^n \int_{x_i^\ast(q_i)}^{x_i} \frac{\partial \mathcal F_i(q_i, s)}{\partial q_i} \, \mathrm ds \, \left( -\mathcal F_i(q_i, x_i) + w_i \right).
  \end{align}
  Since $\mathcal F_i(q_i, x_i)$ is increasing in $x_i$ and $\mathcal F_i(q_i, x_i^\ast(q_i)) = w_i$, it holds that $\nabla_q^T \mathcal L(q,x) \left( -\mathcal F(q,x) + w \right) \leq 0$ where the equality holds if and only if $\mathcal F_i(q_i, x_i) = w_i$ for all $i$ in $\{1, \cdots, n\}$. This completes the proof.
\end{proof}

\subsubsection{$\delta$-Passivity of EDM \eqref{eq:edm}} \label{sec:delta_passivity_edm}
Let us define $\mathcal V_i(p,x) = \sum_{j=1}^n x_j \rho_{ji}(p) - x_i \sum_{j=1}^n \rho_{ij}(p)$. We refer to EDM \eqref{eq:edm} \textit{$\delta$-passive} if there is a continuously differentiable non-negative function $\mathcal S: \mathbb R^n \times \mathbb X \to \mathbb R_+$ satisfying 
\begin{subequations} \label{eq:delta_passivity_requirement_01}
  \begin{align}
    &\nabla_p \mathcal S(p,x) =  \mathcal V(p,x) \\
    &\nabla_x^T \mathcal S(p,x) \mathcal V(p,x) \leq 0,
  \end{align}
\end{subequations}
and
\begin{align} \label{eq:delta_passivity_requirement_02}
  \mathcal S(p,x)=0 &\iff \mathcal V(p,x)=0 \nonumber \\
  &\iff \nabla_x^T \mathcal S(p,x) \mathcal V(p,x) = 0,
\end{align}
where $\mathcal V = (\mathcal V_1, \cdots, \mathcal V_n)$.
We refer to $\mathcal S$ as the \textit{$\delta$-storage function}. The EDM defined using the Smith learning rule is $\delta$-passive and has the $\delta$-storage function given by $\mathcal S(p,x) = \frac{\varrho}{2} \sum_{i,j=1}^n x_j [p_i-p_j]_+^2$.

\subsection{Proof of Theorem~\ref{thm:convergence}} \label{sec:theorem_proof}
Recall that the agents are revising their strategy selections using their own payoff vector estimates. The state equation for the closed-loop system consisting of \eqref{eq:task_allocation_games} and \eqref{eq:edm_with_error_term} is given as
\begin{subequations} \label{eq:tag_edm}
  \begin{align}
    \dot q_i(t) &= - \mathcal F_i(q_i(t), x_i(t)) + w_i \label{eq:tag_in_closed_loop} \\
    \dot x_i(t) &= \lambda \mathcal V_i(p(t), x(t)) + \lambda \xi_i(t), \label{eq:edm_in_closed_loop}
  \end{align}
\end{subequations}
where $\mathcal V_i(p, x) = \sum_{j=1}^n x_j \rho_{ji}(p) - x_i \sum_{j=1}^n \rho_{ij}(p)$ and $p(t) = q(t)$. Since \eqref{eq:edm} is assumed to be $\delta$-passive, when $\xi_i(t) = 0, ~ \forall i \in \{1, \cdots, n\}$, \eqref{eq:edm_in_closed_loop} is $\delta$-passive and has the $\delta$-storage function $\mathcal S(p,x)$ as explained in Appendix~\ref{sec:delta_passivity_edm}. Also, according to Lemma~\ref{lemma:delta_antipassivity}, \eqref{eq:tag_in_closed_loop} is $\delta$-antipassive with the $\delta$-antistorage function $\mathcal L(q,x)$ as in \eqref{eq:antidelta_storage_function}.

We provide a two-part proof: in the first part, we show that $\mathcal L(q(t),x(t))$ becomes arbitrarily small as we select small $\lambda$, and in the second part, we argue that $\mathcal S(p(t),x(t))$ also becomes arbitrarily small. Then, based on \eqref{eq:delta_antipassivity_requirements_02} and \eqref{eq:delta_passivity_requirement_02}, we conclude that $(q(t), x(t))$ converges to the equilibrium $(q^\ast, x^\ast)$ as $\lambda$ becomes arbitrarily small.

We first evaluate the time-derivative of $\lambda \mathcal S(p(t), x(t)) + \mathcal L(q(t), x(t))$.
\begin{align} \label{eq:derivative_of_lyapunov_t}
  &\frac{\mathrm d}{\mathrm d t} \left( \lambda \mathcal S(p(t), x(t)) + \mathcal L(q(t), x(t)) \right) \nonumber \\
  &= \underbrace{\lambda \nabla_p^T \mathcal S(p(t), x(t))}_{=\dot x^T(t) - \lambda \xi^T(t)} \dot p(t) \nonumber \\
  &\qquad + \lambda \nabla_x^T \mathcal S(p(t), x(t)) \underbrace{\dot x(t)}_{=\lambda \mathcal V(p(t), x(t)) + \lambda \xi(t)} \nonumber \\
  &\qquad + \nabla_q^T \mathcal L(q(t), x(t)) \dot q(t) + \underbrace{\nabla_x^T \mathcal L(q(t), x(t))}_{=-\dot q^T(t)} \dot x(t), \nonumber \\
  &= \lambda \xi^T (t) \left( \lambda \nabla_x \mathcal S(p(t), x(t)) - \dot p(t) \right)  \nonumber \\
  &\quad + \underbrace{\lambda^2 \nabla_x^T \mathcal S(p(t), x(t)) \mathcal V(p(t), x(t))}_{\leq 0} \!+\! \underbrace{\nabla_q^T \mathcal L(q(t), x(t)) \dot q(t)}_{\leq 0},
\end{align}
where $\mathcal V = (\mathcal V_1, \cdots, \mathcal V_n)$ and $\xi = (\xi_1, \cdots, \xi_n)$. First, we note that since $q(t)$ and $x(t)$ are bounded, by the continuous differentiability of $\mathcal S$, we have that $\nabla_x \mathcal S(p(t), x(t))$ is bounded. Also, recall that $\xi (t)$ and $\dot p(t)$ are bounded. Hence, for any $\delta > 0$, we can select sufficiently small $\lambda$ for which $\lambda \xi^T (t) \left( \lambda \nabla_x \mathcal S(p(t), x(t)) - \dot p(t) \right) < \delta$ holds. Consequently, for sufficiently small $\lambda$, we can establish
\begin{align*}
  &\frac{\mathrm d}{\mathrm d t} \left( \lambda \mathcal S(p(t), x(t)) + \mathcal L(q(t), x(t)) \right) \nonumber \\
  &<\! \delta \!+\! \lambda^2 \nabla_x^T \mathcal S(p(t),\! x(t)) \mathcal V(p(t),\! x(t)) \!+\! \nabla_q^T \mathcal L(q(t),\! x(t)) \dot q(t).
\end{align*}
From the above inequality, we can observe that the function $\lambda \mathcal S(p(t), x(t)) + \mathcal L(q(t), x(t))$ decreases if $\lambda^2 \nabla_x^T \mathcal S(p(t), x(t)) \mathcal V(p(t), x(t)) + \nabla_q^T \mathcal L(q(t), x(t)) \dot q(t) < -2\delta$ holds. Since $\lambda \mathcal S(p(t), x(t)) + \mathcal L(q(t), x(t))$ is a non-negative function, for some $t > 0$,
it should hold that
\begin{multline} \label{eq:condition_for_entering_invariant_set}
  \lambda^2 \nabla_x^T \mathcal S(p(t), x(t)) \mathcal V(p(t), x(t)) \\ + \nabla_q^T \mathcal L(q(t), x(t)) \dot q(t) \geq - 2\delta.
\end{multline}
Otherwise, $\lambda \mathcal S(p(t), x(t)) + \mathcal L(q(t), x(t))$ eventually becomes negative. Let $\mathbb U_\delta$ be a set defined as
\begin{multline*}
  \mathbb U_\delta = \{ (q,x) \in [0, q_{\text{max}}]^n \times \mathbb X \,|\, \\ \nabla_q^T \mathcal L(q, x) (-\mathcal F(q,x) + w) \geq -2\delta \}.
\end{multline*}
Also, define a positive constant $\alpha$ using $\mathbb U_\delta$ as 
\begin{align*}
  \alpha = \sup_{(q,x) \in \mathbb U_\delta} \left( \lambda \mathcal S(p, x) + \mathcal L(q, x) \right).
\end{align*}
Notice that, by construction, the set $\{(q,x) \in [0, q_{\text{max}}]^n \times \mathbb X \,|\, \lambda \mathcal S(p, x) + \mathcal L(q, x) \leq \alpha\}$ is invariant, and according to \eqref{eq:condition_for_entering_invariant_set}, the trajectory $(q(t), x(t))$ always enters the set. Hence, we conclude that for some large $T_\lambda > 0$, it holds that $\mathcal L(q(t), x(t)) \leq \alpha, ~ \forall t \geq T_\lambda$.

Since $\alpha$ decreases as does $\lambda$, we can construct a class $\mathcal K$ function\footnote{See \cite[Chapter~4.4]{Khalil:1173048} for its definition.} $\kappa_1$ of $\lambda$ for which it holds that
\begin{align}
  \mathcal L(q(t), x(t)) < \kappa_1 (\lambda), ~ \forall t \geq T_\lambda.
\end{align}
Hence, by \eqref{eq:delta_antipassivity_requirements_02} and Assumption~\ref{assumption:estimation_convergence}, we can select small $\lambda$ to make $\limsup_{t \to \infty} \|\dot q(t)\|_2$ and $\limsup_{t \to \infty} \| \xi (t) \|_2$ arbitrarily small.

Now suppose, for given $\epsilon > 0$, we select $\lambda$ such that $\| \dot q(t) \|_2 < \epsilon, ~ \forall t \geq T_\lambda$ holds. Then, for each $i$ in $\{1, \cdots, n\}$, it holds that
\begin{align} \label{eq:bound_on_p_dot}
  \left| \mathcal F_i(q_i(t), x_i(t)) - w_i \right| < \epsilon.
\end{align}
if $q_i(t) < q_{\text{max}}$.
Otherwise $q_i(t) = q_{\text{max}}$.

By the implicit function theorem and Assumption~\ref{assumption:function_F}, there exists a continuously differentiable function $\bar q_i = g_i(x_i)$ for which $\mathcal F_i(\bar q_i, x_i) = w_i$ holds if such $\bar q_i \in [0, q_{\text{max}}]$ exists for a given $x_i \in [0, 1]$. We then extend the domain of the function $g_i$ to the entire $[0, 1]$ by setting $q_{\text{max}} = g_i(x_i)$ if $\mathcal F_i(q_{\text{max}}, x_i) < w_i$.\footnote{Note that the extended $g_i$ is continuously differentiable except at $x_i$ for which $\mathcal F_i(q_{\text{max}}, x_i) = w_i$ holds.}
Note that $g_i$ is a non-increasing function of $x_i$, i.e., $\frac{\partial g_i}{\partial x_i} \leq 0$.
Consequently, by Assumption~\ref{assumption:function_F}, if $\epsilon$ is sufficiently small,
$q_i(t)$ lies in a small neighborhood of $g_i(x_i(t))$.

Let us define
\begin{align}
  \mathcal G(x) =
  \begin{pmatrix}
    g_1(x_1) &
    g_2(x_2) &
    \cdots &
    g_n(x_n)
  \end{pmatrix}^T,
\end{align}
and let $\zeta(t) = q(t) - \mathcal G(x(t))$.
Eq. \eqref{eq:edm_in_closed_loop} can be re-written as
\begin{align}
  \dot x_i(t)
  &= \lambda \mathcal V_i(\mathcal G(x(t)) + \zeta(t), x(t)) + \lambda \xi_i(t) \nonumber \\
  &= \lambda \mathcal V_i(\mathcal G(x(t)), x(t)) + \lambda \bar{\xi}_i(t),
\end{align}
where $\bar{\xi}_i(t)$ is defined as
\begin{align*}
  \bar{\xi}_i(t)
  &= \xi_i(t) \!+\! \sum_{j=1}^n x_j(t) \big( \rho_{ji}(\mathcal G(x(t)) \!+\! \zeta(t)) \!-\! \rho_{ji}(\mathcal G(x(t))) \big) \nonumber \\
  &\quad - x_i(t) \sum_{j=1}^n \big( \rho_{ij}(\mathcal G(x(t)) + \zeta(t)) - \rho_{ij}(\mathcal G(x(t))) \big).
\end{align*}
By the Lipschitz continuity of $\rho_{ji}$, it holds that
\begin{align*}
  &\Bigg| \sum_{j=1}^n x_j(t) \big( \rho_{ji}(\mathcal G(x(t)) + \zeta(t)) - \rho_{ji}(\mathcal G(x(t))) \big) \nonumber \\
  &\quad - x_i(t) \sum_{j=1}^n \big( \rho_{ij}(\mathcal G(x(t)) + \zeta(t)) - \rho_{ij}(\mathcal G(x(t))) \big) \Bigg| \nonumber \\
  &\leq c \left\| \zeta(t) \right\|_2.
\end{align*}
Therefore, if $\lambda$ is sufficiently small, both $\|\xi(t)\|_2$ and $\|\zeta(t)\|_2$ eventually become arbitrarily small, as does $|\bar \xi_i(t)|$.

Now, we evaluate the time derivative of $\mathcal S(\mathcal G(x(t)), x(t))$ as follows.
\begin{align}
  &\frac{\mathrm d}{\mathrm d t} \mathcal S(\mathcal G(x(t)), x(t)) \nonumber \\
  &=\!\nabla_p^T\! \mathcal S(\mathcal G(x(t)), \!x(t)) \nabla \mathcal G(x(t)) \dot x(t) \!+\! \nabla_x^T\! \mathcal S(\mathcal G(x(t)), \!x(t)) \dot x(t) \nonumber \\
  &\leq \!\lambda \left( \mathcal V^T (\mathcal G(x(t)), x(t)) \nabla \mathcal G(x(t)) \!+\! \nabla_x^T \mathcal S(\mathcal G(x(t)), x(t)) \right) \bar{\xi}(t) \nonumber \\
  &\qquad + \lambda \nabla_x^T \mathcal S(\mathcal G(x(t)), x(t)) \mathcal V(\mathcal G(x(t)), x(t)),
\end{align}
where we used \eqref{eq:delta_passivity_requirement_01} and the fact that $\nabla \mathcal G(x)$ is a negative semi-definite matrix for every $x \in \mathbb X$. Since $\bar \xi(t) = (\bar \xi_1 (t), \cdots \bar \xi_n (t))$ eventually vanishes as $\lambda$ decreases, using the same arguments as in the first part of the proof, we can conclude that there is a $\mathcal K$ class function $\kappa_2$ of $\lambda$ for which
\begin{align}
  \mathcal S(\mathcal G(x(t)), x(t)) < \kappa_2(\lambda), ~ \forall t \geq T_{\lambda}'
\end{align}
holds for some $T_\lambda' > 0$. Therefore, we conclude that by decreasing $\lambda$, we can guarantee the convergence of $(q(t), x(t))$ to $(q^\ast, x^\ast)$. This completes the proof.

\balance
\bibliographystyle{IEEEtran}
\bibliography{IEEEabrv,references}

\end{document}